\documentclass[11pt]{article}
\usepackage{amsmath}
\usepackage{amssymb}
\usepackage{amsthm}
\usepackage{enumerate}
\usepackage{ifthen}
\usepackage[margin=1in]{geometry}
\usepackage{url}
\usepackage{graphicx} 

\theoremstyle{plain}

\DeclareMathSymbol{\flecha}{\mathbin}{AMSa}{"4C}

\newcommand{\G}{\mathcal G}

\newtheorem{theorem}{Theorem}

\newtheorem{definition}{Definition}
\newtheorem{lemma}{Lemma}
\newtheorem{proposition}{Proposition}

\newcounter{example}
                        {\hfill\scriptsize$\Box$}

\newcommand{\denselist}{\topsep 0pt\itemsep 0pt}
\newcommand{\tup}[1]{\langle #1 \rangle}

\newcommand{\set}[1]{\{#1\}}
\newcommand{\setdef}[2]{\{#1\; : \; #2 \}}

\renewcommand{\models}{\vDash}

\newcommand{\struc}{\text{Struc}}

\newcommand{\FO}{\text{FO}}

\newcommand{\A}{\mathcal{A}}
\newcommand{\B}{\mathcal{B}}

\newcommand{\qr}{\text{q.r.}}

\title{A sufficient condition for first order non-definability of arrowing problems}
\author{Nerio Borges \\
        Departamento de Matem\'aticas \\
        Universidad Sim\'on Bol\'{\i}var \\ 
        Caracas, Venezuela \\ 
        \url{nborges@usb.ve} 
}

\begin{document}

\maketitle

\abstract{We here present a sufficient condition for general arrowing problems
to be non definable in first order logic, based in well known
tools of finite model theory e.g. Hanf's Theorem and
known concepts in finite combinatorics, like senders and determiners.}

\section{Introduction}

{\sc Arrowing} is the problem of deciding, given three finite, undirected, simple graphs $F,G,H$
if for every coloring of the edges of $F$ with two colors (e.g. red and blue) 
a red $G$ or a blue $H$ occurs. If it is the case,
we write $F\rightarrow (G,H)$. If not, then we write $F\nrightarrow (G,H)$.

If we let $G,H$ range in a class of graphs $\Omega$, 
then we can denote the restricted resulting problem as
$\text{\sc Arrowing}_\Omega$. 
We can even fix the graphs $G,H$ and, given a graph $F$, 
ask whether $F\rightarrow (G,H)$ or not. Denote this problem as
$\text{\sc Arrowing}(G,H)$. 
Thus $\text{\sc Arrowing}_\Omega(G,H)$ is the problem of deciding,
given a graph $F$, whether $F\rightarrow (G,H)$ or not for a pair
of fixed graphs $G,H$ in $\Omega$ ($G,H$ must be given as a part of the input).

The complexity of {\sc Arrowing} has been widely studied. 
Some arrowing problems are
known to be complete via polynomial many-one reductions in complexity
classes like P, NP \cite{Burr} and $\Pi_2^p$ \cite{Schaefer},
and the problem $\text{\sc Arrowing}(G,H)$, also known as 
the {\sc Monochromatic Triangle} has been proved NP complete via
first order reductions \cite{medina:thesis}.

We here present a sufficient condition for general arrowing problems
to be non definable in first order logic, based in well known
tools of finite model theory e.g. Hanf's Theorem and
known concepts in finite combinatorics.

\section{Preliminaries}\label{sec:preliminaries}

This section is an attempt to keep this work self-contained. 
the subsection \ref{subs:graphs} deals with graphs and introduces a non-standard
notation for some graph operations.

\subsection{Preliminaries in logic}

A {\em vocabulary} is a tuple of symbols
\[
\tau=\tup{R_1^{a_1},R_2^{a_2},\ldots,f_1^{b_1},f_2^{b_2}\ldots,c_1,c_2\ldots}
\]
where
each $R_j$ is a relational symbol of arity $a_j$,
each $f_k$ is a function symbol of arity $b_j$
and each $c_i$ is a
a constant symbol. 

If $\tau$ has no function symbols, we call it a {\em relational} vocabulary. 
A vocabulary is {\em finite} if it consists of a finite set of symbols.
From now on the greek letters $\tau$ and $\sigma$ will denote finite relational vocabularies.

A structure for $\tau$, also called a $\tau$-structure, is a tuple 
$\A=\tup{|\A|,R_1^\A,\ldots,R_r^\A,c_1^\A,\ldots,c_s^\A}$
where $|\A|$ is the universe (or domain) of $\A$, each $R_j^\A\subseteq|\A|^{a_j}$
is a $a_j$-ary relation over $|\A|$, and each $c_j\in|\A|$ is an
element of $|\A|$. 

For vocabulary $\tau$, $\struc(\tau)$ denotes the class of all
finite structures with size $\|\A\|\geq2$, i.e.\ structures whose
universe is an initial segment $[n]=\{0,1,\ldots,n-1\}$ 
of the set $\mathbb N$ of the natural numbers with $n\geq2$.
We consider here only finite structures.

The {\em language} $\FO(\tau)$ is the set of all
well-formed first order formulas over the vocabulary $\tau$.

If $\tau$ is relational, then its {\em terms} are either first order variables
or constants symbols from $\tau$.
An {\em atomic} formula over vocabulary $\tau$ has the form
$P(t_1,\ldots t_k)$ with $P$ a $k$-ary relational symbol and $t_1,\ldots,t_k$ are terms.
A {\em literal} is an atomic formula (and then we say it is {\em positive}) 
or the negation of an atomic formula (and then we say it is {\em negative}).

\subsection{Graphs}\label{subs:graphs}

A {\em graph} is a structure for the vocabulary
$\sigma=\tup{E}$ consisting of one binary relation $E$ 
i.e. a pair $\A=\tup{|\A|, E^\A}$ where $|\A|$ is an initial
segment of $\mathbb N$ called the set of {\em vertices} of $\A$, and
$E^\A$ is a subset of $|\A|^2$. 
Tipically, graphs are denoted by latin capital letters as $G,F,H$.
When we consider a graph $G$, we often denote its vertex set as $V_G$
and the set of all its edges as $E_G$.
A simple, undirected graph $G$ is a graph where the relation
$E(G)$ is irreflexive and symmetric.

We need to define two binary operations between graphs. Intuitively, the idea
is to ``join'' both graphs together identifying two edges.
\begin{definition}\label{de:jointgraphs}
Given a graph $G$ with a distinguished edge $(a,b)$ and a graph $H$ with distinguished edge $(c,d)$, 
we define the graph 
$F=G(a,b)\oplus (c,d)H$, the result of identifying
edges $(a,b)$ and $(c,d)$ (also identifying vertices $a$ with $c$ and $b$ with $d$) in the following way:
\begin{itemize}\denselist
\item The set of vertices is the disjoint union of $V_G$ and $V_H$ 
	without the vertices representing $a$ and $b$:
	\[
	V_F=((V_G\times\set{0})\cup(V_H\times\set{1}))-\set{(a,0),(b,0)}
	\]
\item The set of edges remains the same for $H$; as for the $G$ part, any edge incident
	in $a$ will be now incident in our copy of $c$, and any edge incident in $b$ will
	be now incident in our copy of $d$:
	\begin{alignat*}{2}
	E_F = & \setdef{\tup{(u,0),(v,0)}}{u,v\not\in\set{a,b},(u,v)\in E_G}
				\cup \setdef{\tup{(u,1),(v,1)}}{(u,v)\in E_F}\\
			& \cup \setdef{\tup{(u,0),(c,1)},\tup{(c,1),(u,0)}}{(u,a)\in E_G}\\
			& \cup \setdef{\tup{(u,0),(d,1)},\tup{(d,1),(u,0)}}{(u,b)\in E_G}
	\end{alignat*}
\end{itemize}
\end{definition}
In figure \ref{fig:example-ident}, $F$ is the graph obtained when one identifies edge $(1,2)$
in $G$ with edge $(1,2)$ in $H$.

\begin{figure}
\begin{center}  
\includegraphics[height=3in,width=4in,angle=0]{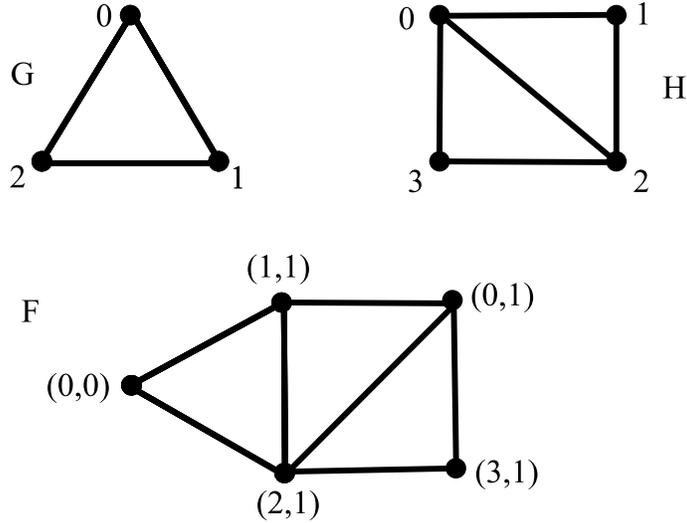}  
\caption{\small \sl $F=G(2,1)\oplus(2,1) H$\label{fig:example-ident}}  
\end{center}  
\end{figure} 

\begin{lemma}\label{le:associativity}
Suppose $F_1,F_2$ and $F_3$ are graphs. If $F_1$ has a distinguished edge $(a_1,b_1)$, 
$F_2$ has two different distinguished edges $(a_2,b_2)$ and $(a_3,b_3)$,
and $F_3$ has distinguished edge $(a_4,b_4)$
respectively, then:
\begin{equation}\label{eq:assoc}
F_1(a_1,b_1)\oplus (a'_2,b'_2)\left[F_2(a_3,b_3)\oplus (a_4,b_4)F_3\right]\cong
\left[F_1(a_1,b_1)\oplus (a_2,b_2)F_2\right](a'_3,b'_3)\oplus (a_4,b_4) F_3
\end{equation}
Where $(a'_2,b'_2)$ is the edge corresponding to $(a_2,b_2)$ in $F_2(a_3,b_3)\oplus (a_4,b_4)F_3$
and $(a'_3,b'_3)$ is the edge corresponding to $(a_3,b_3)$ in $F_1(a_1,b_1)\oplus (a_2,b_2)F_2$. 
\end{lemma}

\begin{proof}
Straight forward.
\end{proof}

 Lemma \ref{le:associativity} says that $\oplus$ is associative in some sense. 
 As a consequence of associativity, this notation is unambiguous:
\[
F_1(a_1,b_1)\oplus (a_2,b_2)F_2(a_3,b_3)\oplus (a_4,b_4)F_3
\] 
We use here the notation $(a_2,b_2)$ (for instance)
instead of $(a'_2,b'_2)$ as in equation \ref{eq:assoc}
because the late is innecesarily cumbersome.

We can also identify two different edges of the same graph.
Notice the following is not a binary operation over graphs:
\begin{definition}\label{de:ident}
If $G$ is a graph and $(a,b),(c,d)$ are two of its edges, we define
the graph $F=G[(a,b)\sim (a',b')]$ as follows:
\begin{itemize}\denselist
\item $V_F$ is the set of blocks in the following partition of $V_G$:
	\[
	V_G=\set{a,a'}\cup\set{b,b'}\cup\setdef{\set{u}\in \wp(V_G)}{u\not\in\set{a,b,a',b'}}
	\]
\item $E_F=\setdef{([u],[v])\in V_F^2}{(u,v)\in E_G}$ being $[u]$ the equivalence class of $u$
	for all $u\in V_G$.
\end{itemize}
\end{definition}

In figure \ref{fig:example-ident2}, $F$ is $G[(2,3)\sim (0,6)]$, the result of
the identification of edge $(2,3)$ with edge $(0,6)$

\begin{figure}
\begin{center}  
\includegraphics[height=3.25in,width=4.25in,angle=0]{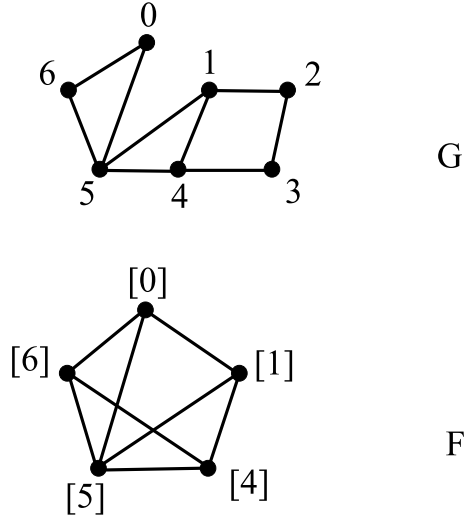}  
\caption{\small \sl $F=G[(2,3)\sim (0,6)]$\label{fig:example-ident2}}  
\end{center}  
\end{figure}

\subsection{Senders and determiners}

Definitions in this subsection correspond to the ones given in \cite{Burr}.

A {\em 2-coloring} of the edges of a graph $F=\tup{V_F,E_F}$, is a function
$\mathbf c:E_F\longrightarrow \set{0,1}$. Informally, we will say the edge $(a,b)$
is ``red'' if $\mathbf c(a,b)=0$ and ``blue'' otherwise.

Given two graphs $G,H$ we say that a 2-coloring $\mathbf c$ of the edges of $F$ is
{\em $(G,H)$-good} if there is no red subgraph of $F$ isomorphic to $G$ and 
no blue subgraph of $F$ isomorphic to $H$ according to $\mathbf c$.
We write $F\rightarrow (G,H)$ if $F$ has no $(G,H)$-good colorings
and we write $F\nrightarrow (G,H)$ if the contrary holds. 

A graph $F$ is {\em $(G,H)$-minimal} if $F\rightarrow(G,H)$ but
$F'\nrightarrow(G,H)$ for every graph $F'$ properly contained in $F$. 
The class of all $(G,H)$-minimal graphs is denoted as $\mathcal R(G,H)$.

If $G$ and $H$ are two graphs, a {\em $(G,H,f)$-determiner} or
simply a $(G,H)$-determiner is a graph $F$ with a special edge $f$
such that:
\begin{enumerate}\denselist
\item There is a $(G,H)$-good coloring for $F$, and
\item in every $(G,H)$-good coloring, $f$ is always red.
\end{enumerate}
We then say that $f$ is the {\em signal edge} of $F$.
On the other hand, a graph $F$
with special edges $e$ and $f$ is a {\em negative $(G,H,e,f)$-sender} if
\begin{enumerate}\denselist
\item There is a $(G,H)$-good coloring for $F$, 
\item in every $(G,H)$-good coloring, $e$ and $f$ have different colors, and\label{neg-sender-different}
\item $F$ has a $(G,H)$-good coloring where $e$ is red and another one where $e$ is blue.
\end{enumerate}
If we change condition \ref{neg-sender-different} to:
\begin{quote}
`In every $(G,H)$-good coloring, $e$ and $f$ have the same color.'
\end{quote}
then $F$ is a {\em positive $(G,H,e,f)$-sender}.

If $F$ is either a positive or negative $(G,H,e,f)$-sender we will say that
$e$ and $f$ are the {\em signal edges} of $F$.
When we are referring to senders (whether they are positive or negative), we will just write
$(G,H)$-sender instead of $(G,H,e,f)$-sender.

\begin{definition}
A negative (positive) $(G,H)$-sender $F$ is {\em minimal} if $F'$ is not a negative (positive) $(G,H)$-sender
for every $F'\subset F$. 
\end{definition}

Our next result is related to this concept.

\begin{lemma}
If a pair of graphs $G,H$ has a negative (resp. positive) $(G,H)$-sender,
then it has a minimal negative (resp. positive) $(G,H)$-sender
\end{lemma}

\begin{proof}
This follows from the fact that graphs can be well ordered and the fact
that the set of negative (resp. positive) $(G,H)$-senders is non-empty.
\end{proof}

The existence of senders and determiners for some families of pairs of graphs is stablished in
\cite{Burr}.

\subsection{First Order equivalence between structures}

We want to give a sufficient condition for arrowing problems to be not first order definable.
A well known strategy to prove non definability in first order for some problem $A$ is, given any $r\in\mathbb N$,
showing that no formula with $r$ nested quantifiers can define $A$.
The number of nested quantifiers in a formula is called its {\em quantification rank} (q.r.). 
Formally, we can define it inductively \cite{Ebbinghaus-Flum}:
\begin{enumerate}\denselist
\item $\qr (\phi)=0$ for every atomic formula $\phi$,
\item $\qr(\neg\phi)=\qr(\phi)$,
\item $\qr(\phi\land\psi)=\qr(\phi\lor\psi)=\max\set{\qr(\phi),\qr(\psi)}$ and
\item $\qr(\forall x\phi)=\qr(\exists x\phi)=\qr(\phi)+1$ 
		(provided $x$ is a first order variable).
\end{enumerate}

We say that two finite $\tau$-structures are {\em FO-$r$-equivalent} if 
\[
\A\models\phi \iff \B \models\phi
\]
for every $\FO(\tau)$ sentence $\phi$ with $\qr(\phi)\leq r$. 
If $\A$ and $\B$ are FO-$r$-equivalent we write $A\equiv_r^{\FO}\B$.

Hence the strategy mentioned above about non definability in FO is formally stablished
in the following Proposition:
\begin{proposition}\cite{immerman:book}\label{prop:elem-equiv}
Let $\Pi$ be a subset of finite $\tau$-structures. 
Suppose that, for every natural number $r$, there are two finite $\tau$-structures
$\A\in\Pi$ and $\B\in\struc(\tau)- \Pi$ such that $\A\equiv_r^{\FO}\B$.
Then $\Pi$ is not first order definable.
\end{proposition}

One way to prove FO-$r$-equivalence between structures, is via the Hanf's Theorem.
It states that two structures are FO-$r$-equivalent if they are, in a sense,
locally isomorphic.

Suppose $\tau$ is a vocabulary and $\A$ is a finite $\tau$-structure.
We define the {\em Gaifman graph} corresponding to $\A$ as the undirected graph
$\G_\A=\tup{V,E}$ where:
\begin{itemize}\denselist
\item $V=|\A|$, and
\item $(a,b)\in E$ if and only if there is a $k$-ary relation $R$ in $\tau$ and a
	$k$-tuple $(a_1,\ldots, a_k)\in R^\A$ such that $a=a_i, b=a_j$ for some pair $i,j\in\set{1,\ldots, k}$.
\end{itemize}
If $\A$ is an undirected graph, for instance, $\G_\A$ and $\A$ are the same.

Given two elements $a,b$ in the universe of $\A$, the {\em distance} $d(a,b)$
is defined as the length of the shortest path joining $a$ and $b$ in $\G_\A$. If they are in 
different connected components of $\G_\A$ then we define $d(a,b)=\infty$.
For each $a\in|\A|$ we define the {\em $r$-ball centered at a} as the set:
\[
B^\A(a,r):=\setdef{b\in|\A|}{d(a,b)\leq r}
\]

If $c$ is a constant not in $\tau$, and $\tau^*=\tau\cup\set{c}$, we define the 
{\em $r$-neighborhood of $a$ in $\A$} as the finite $\tau^*$-structure
\[
N^\A(a,r):=\tup{\A\upharpoonright_{B^\A(a,r)},a}
\]
The {\em isomorphism type} of a structure $\A$ is the set of all literals
satisfied by $\A$. The isomorphism type of $N^\A(a,r)$ is
the {\em $r$-type of $a$ in $\A$}. 
If, for instance, $\A$ is a graph and $v$ is one of its vertices,
then the $r$-type of $v$ is the set of atomic formulas and negations of atomic formulas
describing the edges of the substructure $N^\A(a,r)$.
Given two $\tau$-structures $\A$ and $\B$, an element $a\in|\A|$ has the same 
$r$-type as an element $b\in|\B|$ if there is an isomorphism $f$ between 
$N^\A(a,r)$ and $N^\B(b,r)$ such that $f(a)=b$.

Denote by $|\A|_\Delta$ the subset of elements of $|\A|$ with $r$-type $\Delta$.
We say that $\A$ and $\B$ are {\em $r$-equivalent} if 
there is a bijection $f:||\A||\longrightarrow ||\B||$ such that the $r$-type of $a$
is the same as the $r$-type of $f(a)$ for all $a\in|\A|$ i.e. if 
$|\A|_\Delta$ and $|\B|_\Delta$ have the same cardinality.

\begin{theorem}[Hanf's Theorem]\cite{Ebbinghaus-Flum}
Suppose $\A,\B$ are two finite $\tau$-structures and $r>0$ is a natural number.
If $\A$ and $\B$ are $2^r$-equivalent, then $\A\equiv_r^\FO\B$.
\end{theorem}

\section{First Order definability}

We present the main result in this section.

\begin{definition}
Let $\Omega$ be a class of graphs. We say that $\Omega$ {\em has
negative (positive) senders with non-adjacent signals}
if for every pair of graphs $(G,H)$ of $\Omega$
there is a negative (positive) $(G,H)$-sender such that
its signal edges have no common vertex.
\end{definition}

\begin{definition}\cite{Berge}
We say that a graph is $k$-connected if it remains connected when we remove any set of $k-1$ vertices, but
gets disconnected if we remove $k$ vertices.
\end{definition}

\begin{lemma}\label{le:main} 
Suppose $\Omega$ is a class of $k$-connected graphs with $k\geq 2$ which has negative senders with 
non-adjacent signals. 

If $n>2$ is a natural number and $G,H$ is a pair of graphs in $\Omega$,
then there is a $(G,H)$-minimal graph $F$ with a pair of vertices $u$ and $v$
such that $d(u,v)\geq n$. 
\end{lemma}

\begin{proof}
Let $F'$ be a minimal negative $(G,H,e,f)$-sender, 
such that its signals are not adjacent. 
Consider a pair of vertices $a$ and $b$ incident with $e$ and $f$ respectively.

Suppose that $x$ is any edge of $F'$ different from $e$ and $f$. 
It is easy to see that the graph $\tilde F$, obtained by removing the edge $x$ from $F'$,
has $(G,H)$-good colorings. Moreover, there is a $(G,H)$-good coloring for $F'$ such
that $e$ and $f$ have the same color, because of the minimality of $F'$ as a negative sender.

Consider $2n+1$ copies $F_1',F_2',\ldots, F'_{2n+1}$ of $F'$. 
Denote the signal edges and the distinguished vertices of $F'_i$ as $e_i,f_i$ and $a_i,b_i$ respectively.
Without loss of generality we can assume that $e_i=(a_i,u_i)$ and $f_i=(b_i,v_i)$
for every $1\leq i\leq 2n+1$. 

Now iterate the operation given in Definition \ref{de:jointgraphs} and form the graph
\[
F''=F'_1 (b_1,v_1) \oplus (a_2,u_2) F'_2 (b_2,v_2) \oplus (a_3,u_3)F'_3\ldots 
F'_{2n}(b_{2n},v_{2n})\oplus(a_{2n+1},u_{2n+1})F'_{2n+1}
\]
We can think of $F'$ as a ``chain'' with $2n+1$ ``links'', where each copy of $F$ is a link joined to the
following link by a signal edge.

Notice that $F''$ is still a negative sender with signal edges $e_1$ and $f_{2n+1}$, because:
\begin{enumerate}[i)]\denselist
\item Every copy of $F'$ has $(G,H)$-good colorings,
\item $e_1$ and $f_{2\ell+1}$ have different colors in any $(G,H)$-good coloring, for $1\leq\ell\leq n$ and
\item no ``new'' copies of $G$ and $H$ are formed when we link copies of $F'$ together, because 
	$F$ and $G$ are at least 2-connected, so $F''$ has $(G,H)$-good colorings.
\end{enumerate}
Also notice that, since $e_j$ and $f_j$ are non-adjacent, $d(a_j,b_j)\geq 1$ thus $d(a_1,b_{2n+1})\geq 2n+1$
(we do not follow notation given for vertices in Definition \ref{de:jointgraphs} since there is no risk of confusion).

Now, as in Definition \ref{de:ident}, form the graph $F=F''[(a_1,u_1)\sim(b_{2n+1},v_{2n+1})]$.
In $F$, call $u$ to the vertex given by $[a_1]$ and call $v$ to the vertex given by $[a_n]$.
It is easy to see that $u$ is in a cycle with length at least $2n+1$ and $d(u,v)\geq n$.

We want to show that $F\in \mathcal R(G,H)$. 
First, notice that $F$ can not have any $(G,H)$-good 
coloring. If it were not the case, then we would have a $(G,H)$-good coloring
for $F''$ where $e_1$ and $f_{2n+1}$ have the same color.
Secondly, suppose we delete an edge $x$ from $F$. Then we must delete it from a copy
of $F'$, say $F'_j$, the $j$-th copy of $F_j$. Thus $F'_j-\set{x}$ is not a negative
sender, due to $F'$ minimality. We are breaking the chain of negative senders formed
by $F''$, so now we have a $(G,H)$-good coloring for $F_j$ where $e_j$ and $f_j$
have the same color and a $(G,H)$-good coloring for $F''$ where $e_1$ and $f_{2n+1}$
have the same color. Therefore, there is a good coloring for $F-\set{x}$ hence $F$
is $(G,H)$-minimal.
 
\end{proof}

\begin{theorem}[Main]\label{th:main}
If $\Omega$ is a class of $k$-connected graphs with $k\geq 2$ which has negative senders
with non-coincident signals, then the class $\text{\sc NonArrowing}_\Omega(G,H)$
is not first order definable for any pair $G,H$ in $\Omega$.
\end{theorem}

\begin{proof}
For an arbitrary natural number $r$, we will use Hanf's Theorem to prove there are two instances of 
$\text{\sc NonArrowing}_\Omega{(G,H)}$, one of them negative and one of
them positive, that can not be distinguished by any FO sentence with quantification rank $r$.

Suppose $F$ is a $(G,H)$-minimal graph with two distinguished vertices $u$ and $v$ such that
the distance between $u$ and $v$ is at least $2^{r+1}$. This graph exists because of
the result in Lemma \ref{le:main}. Now let $F_1$ be $F\sqcup \left(F-\set{u,v}\right)$
and $F_2$ be $\left(F-\set{u}\right)\sqcup \left(F-\set{v}\right)$, where the symbol
$\sqcup$ denotes disjoint union. 
Note that, since $F$ is $(G,H)$-minimal, $F_1$ is a negative instance 
while $F_2$ is a positive instance of 
$\text{\sc NonArrowing}_\Omega{(G,H)}$.

Now, as we want to use Hanf's Theorem, we need to show that $F_1$ and $F_2$
are $2^r$-equivalent 
i.e. that there is a bijection $\phi$
mapping every vertex in $F_1$ to a vertex in $F_2$ with the
same $2^r$-type. We define $\phi$ as follows:
\begin{enumerate}[i.]
\item If $w$ is a vertex in the connected component
	of $F$ isomorphic to $F-\set{u,v}$
	and $d(u,w)\leq 2^r$ in $F$, then $\phi(w)$ is the copy of $w$
	in $F-\set{u}$. We proceed analogously if $d(v,w)\leq 2^r$ in $F$.
	
	If $w$ is neither in the neighborhood $N(u,2^r)$ nor
	in the neighborhood $N(v,2^r)$ of $F$, then $\phi(w)$
	is the copy of $w$ in the connected component
	of $F_2$ which is
	isomorphic to $F-\set{u}$.
\item If $w$ is a vertex in the connected component of $F_1$
	isomorphic to $F$ and $d(u,w)\leq 2^r$ then $\phi(w)$
	is the copy of $w$ in the connected component of $F_2$
	isomorphic to $F-\set{v}$. We proceed in an analogous
	way when $d(v,w)\leq 2^r$.
	
	If $w$ is neither in the neighborhood $N(u,2^r)$ nor
	in the neighborhood $N(v,2^r)$ of $F$, then $\phi(w)$
	is the copy of $w$ in the connected component
	of $F_2$ which is
	isomorphic to $F-\set{v}$.
\end{enumerate}
It is easy to see that this function is bijective and preserves
$2^r$-types. Then, by Hanf's Theorem, both structures are $2^r$-equivalent.
As these construction is possible for every $r\in\mathbb N$, 
we conclude by Proposition \ref{prop:elem-equiv} that 
$\text{\sc NonArrowing}_\Omega(G,H)$ is not first order definable.
 
\end{proof}

\bibliographystyle{plain}
\bibliography{bibliografia}

\end{document}